\documentclass[a4paper, USenglish, cleveref, autoref, thm-restate]{lipics-v2021}

\hideLIPIcs
\nolinenumbers

\usepackage{graphicx} 
\usepackage{hyperref}
\usepackage{todonotes}

\usepackage{amsmath}
\usepackage{amsthm} 
\usepackage{thm-restate}

\title{A Note on Diameter Certification in Trees}

\author{Josef Erik Sedláček}
{Faculty of Information Technology, CTU in Prague, Prague, Czech Republic}
{sedlajo5@fit.cvut.cz}
{0009-0001-7429-2937}
{}

\authorrunning{J.\,E. Sedláček} 


\funding{This work was supported by the Grant Agency of the Czech Technical University in Prague, grant No. SGS23/205/OHK3/3T/18.}

\begin{document}

\maketitle

\begin{abstract}
In the local certification model, certifying the diameter of general graphs requires large certificates, but trees admit more efficient solutions. In this note, we provide a $1$-local certification scheme that certifies whether the diameter of a given tree is at most $d$ using certificates of at most $3\lceil\log_2(d+1)\rceil$ bits.
\end{abstract}


\section{Introduction}

\paragraph*{Local certification and network structure}

This note focuses on \emph{local certification}, a framework in distributed computing used to verify global properties of a network. In this context, the network topology is naturally modeled as a graph $G=(V,E)$, where vertices represent the computing nodes and edges represent the communication links. We are interested in checking global graph properties such as acyclicity, planarity and bounded diameter.

Since vertices typically only have a local view of their immediate neighborhood, most global properties cannot be verified without external assistance. For example, a vertex cannot determine if the entire graph is bipartite just by communicating with its direct neighbors. To overcome this limitation, the local certification model~\cite{Feuilloley21} has been introduced. In this model, a centralized oracle (the \emph{prover}) assigns a label, called a \emph{certificate}, to each vertex. The vertices then communicate with their neighbors up to a constant radius and decide whether to accept or reject the configuration. Returning to the bipartiteness example, the prover could simply assign a color $\{1,2\}$ to each vertex as its certificate. A vertex accepts if all of its neighbors have a certificate of a different color from its own.

A certification scheme is correct if there exists a certificate assignment that makes all vertices accept when the property holds, and for any assignment, at least one vertex rejects when the property does not hold. The primary measure of efficiency in this setting is the certificate size. Specifically, we consider the maximum size of a certificate over all vertices in the graph, measured in bits, and the objective is to minimize this maximum size. The threshold of $\Theta(\log n)$ bits, where $n$ is the number of vertices, has emerged as the standard baseline for compact local certification~\cite{GoosS16}. This size is significant because it enables the certification of structures like spanning trees.

This general verification mechanism can be formalized through various models, most notably \emph{proof-labeling schemes}~\cite{KormanKP10} and \emph{locally checkable proofs}~\cite{GoosS16}. While these settings typically assume that each vertex is equipped with a unique piece of information called an \emph{identifier} of size $O(\log n)$ bits, in this note we operate in the more restrictive \emph{anonymous model}, meaning that vertices do not possess any identifiers. Note that since our goal is to establish an upper bound, presenting a scheme for a model without identifiers yields a stronger result, as the correctness and size bounds immediately carry over to settings where identifiers are available. For a comprehensive overview of these different certification models and their variants, we refer the reader to~\cite{Feuilloley21}.

In this note, our main objective is to certify that the diameter of a given tree is at most a constant $d$. 
The classical approach to local certification often considered general graphs. It has been proven that certifying that the diameter of a general graph is at most $d$ requires certificates of size $\tilde{\Omega}(n)$ bits per vertex~\cite{CENSORHILLEL2020112}. 

Recent research has shifted towards designing optimized certification schemes for restricted graph classes. By leveraging the underlying topology, the required certificate sizes can be reduced. For example, while earlier works established tight bounds for local properties in bipartite graphs~\cite{GoosS16}, recent developments have provided schemes for leader election in chordal and grid graphs whose size is logarithmic in the diameter~\cite{LeaderWithoutNamingIt}, as well as constant-size leader certification in meshed graphs~\cite{ChalopinCK26}. 
Contributing to this line of research, we focus on tree topologies, providing a simple, and compact scheme for diameter verification.

\paragraph*{Meta-theorems and MSO logic}
A highly successful approach in understanding the capabilities of local certification on restricted graph classes has been the establishment of meta-theorems, inspired by Courcelle's theorem in centralized computing. Bousquet, Feuilloley, and Pierron~\cite{FeuilloleyBP22} proved that on tree topologies, \emph{any} graph property expressible in Monadic Second-Order (MSO) logic can be certified with certificates of size $O(1)$ (Theorem 2.2 in~\cite{FeuilloleyBP22}).

This meta-theorem has direct implications for the certification of the diameter of a graph. For any fixed integer $d$, the property that a graph has diameter at most $d$ can be easily expressed in First-Order (FO) logic by quantifying over all pairs of vertices and asserting the existence of a path of length at most $d$ between them. Therefore, as a direct corollary of the meta-theorem by Bousquet et al., certifying that a tree has diameter at most $d$ can be done with $O(1)$ bits with respect to the number of vertices $n$.

\paragraph*{Motivation for our result}
While the aforementioned MSO meta-theorem completely resolves the asymptotic complexity of diameter certification on trees with respect to $n$, the underlying framework is highly generic and complex. Rather than offering a direct solution for this specific problem, the automata-theoretic approach hides the dependency on the parameter $d$ inside constants that grow rapidly with the quantifier depth of the formula.

In this note, we show that one does not need to invoke heavy machinery to certify the diameter on trees. We present a completely explicit, simple, and direct $1$-local certification scheme. 

\paragraph*{Our Contribution}
We contribute to this line of research by providing a direct certification scheme for the diameter of trees, explicitly bounding the certificate size in terms of $d$. Specifically, we establish the following result:
\begin{restatable}{theorem}{TheoremDiamtrees}
\label{thm:diamtrees}
For tree graphs, there exists a $1$-local certification scheme $(f,\mathcal{A})$ of size $3\lceil\log_2 (d+1)\rceil$ that accepts a tree $T$ if and only if the diameter of $T$ is at most~$d$.
\end{restatable} 

Our approach relies on simulating two runs of a Breadth-First Search (BFS) directly encoded within the certificates. By assigning distance values from chosen endpoints of a longest path in the tree, vertices can verify the diameter threshold locally. While local certification results typically focus only on asymptotic complexity, the simplicity of our construction allows us to optimize the scheme down to the exact bit. We demonstrate that a small multiplicative factor of $3$ suffices, providing a precise alternative to the general MSO meta-theorem.

\section{Model, definitions and notation}

All graphs in this paper are undirected, connected, simple graphs, denoted by $G = (V, E)$.
Let $n$ denote the number of vertices and $d$ the diameter of the graph.
Usually, the vertices are assigned \emph{unique identifiers} encoded on $O(\log n)$ bits, but in this paper we assume the \emph{anonymous} version, where no identifiers are assigned to vertices.
Neighbors of a vertex $v$ are denoted as $N_G(v)$, and if $G$ is clear from the context, $N(v)$ is used. Distance between two vertices $u, v$ is denoted as $d_G(u, v)$, and the subscript is omitted if $G$ is clear from the context.
Let $V[v,r]$ denote the set of vertices at distance at most $r$ from $v$, which we call the $r$-local neighborhood of~$v$.

A \emph{graph property} is formally a set of graphs that is closed under isomorphism.
A \emph{certificate assignment} $P$ for $G$ is a function $P\colon V(G) \rightarrow \{0,1\}^*$ that associates with each vertex a \emph{certificate}. 
We say that $P$ has size $s$ if $|P(v)| \leq s(n)$ for every $v$.
A \emph{verifier} is a function that takes as an input a graph $G$, its certificate assignment $P$ and $v \in V(G)$ and outputs either $0$ or $1$.

The subgraph $G[V[v,r]]$ is denoted as $G[v,r]$ and the restriction of~$P$ to $V[v,r]$ is denoted as $P[v,r]$, that is $P[v,r] \colon V[v,r]\to \{0,1\}^*$. 
A verifier $\mathcal{A}$ is $r$-\emph{local} if $\mathcal{A}(G,P,v) = \mathcal{A}(G[v,r],P[v,r],v)$ for all $G$, $P$, and $v$. 

An $r$-local \emph{certification scheme} certifying a property of graphs $\mathcal{P}$ is a pair $(f,\mathcal{A})$,
where $\mathcal{A}$ is an $r$-local verifier and $f$, called the \emph{prover}, assigns to each $G \in \mathcal{P}$ a certificate assignment $P$ such that the following properties hold.
\begin{itemize}
    \item \emph{Completeness}: If $G \in \mathcal{P}$, then $\mathcal{A}(G[v, r], P[v, r], v) = 1$ for all $v$, where $P =f(G)$.
    \item \emph{Soundness}: If $G \notin \mathcal{P}$, then for every certificate assignment $P'$, there is $v$ such that \\ $\mathcal{A}(G[v, r], P'[v, r], v)~=~0$.
\end{itemize}
We say that $(f, \mathcal{A})$ has size $s : \mathbb{N} \to \mathbb{N}$ if $|f(G)(v)| \leq s(|V(G)|)$ for all $G \in \mathcal{P}$ and all $v \in V(G)$.

\section{Diameter certification on trees of size $3\lceil\log_2(d+1)\rceil$}
In this section, we prove the following theorem
\TheoremDiamtrees*

Let us first describe the scheme, namely the prover strategy on yes-instances and the verification at the vertices.    
Let $T$ be a tree with diameter at most $d$ and $P = f(T)$ be the certificate assignment for $T$.

For every $v\in V$, the certificate $P(v)$ is a triple $(d_1, m_1, d_2)$:
\begin{itemize}
    \item $d_1 = d(r_1,v)$, the distance to $r_1$ from $v$, where $r_1$ is an arbitrarily selected unique vertex.
    \item A number $m_1 = \max_{v\in V}(d(v,r_1))$, the maximum distance to vertex $r_1$ across all $v\in V$.
    \item $d_2 = d(r_2,v)$, where $r_2$ is a unique vertex with $d(r_2,r_1) = m_1$ chosen arbitrarily among all the vertices $u\in V$ such that $d(u,r_1) = m_1$. 
\end{itemize}

Let $d_1(v),m_1(v),d_2(v)$ denote the components of $P(v)$ for a given vertex $v$.
Verification at a vertex $v$ acts as a local consistency check.
The algorithm evaluates a series of conditions and if any of these conditions are met, indicating an invalid or inconsistent certificate assignment, the verifier $\mathcal{A}(v)$ immediately rejects.
If the configuration passes all checks without triggering a rejection, the vertex accepts.

Specifically, verification on a vertex $v$ consists of the following steps: 
\begin{enumerate}
     
    \item For $i\in\{1,2\}$ if:
    \begin{enumerate}
        \item $d_i(v) = 0$ and there exists $u\in N(v)$ $d_i(u) \neq 1$, or
        \item $d_i(v) = k$ for some $k$ and there is not exactly one $u\in N(v)$ such that $d_i(u) = k-1$, or
        \item $d_i(v) = k$ for some $k$ and there is $u\in N(v)$ such that $d_i(u) < k-1$ or $d_i(u) > k+1$, or
        \item $d_i(v) > d$
    \end{enumerate}
    the verifier $\mathcal{A}(v)$ rejects.
    \item If there exists $u\in N(v): m_1(v)\neq m_1(u)$, the verifier $\mathcal{A}(v)$ rejects.
    \item If $d_1(v)>m_1(v)$, the verifier $\mathcal{A}(v)$ rejects.
    \item If $d_2(v) = 0$ and $d_1(v) \neq m_1(v)$, the verifier $\mathcal{A}(v)$ rejects.
    \item Otherwise, $\mathcal{A}(v)$ accepts.
\end{enumerate}

\begin{proof}[Proof of Theorem \ref{thm:diamtrees}]
    First, we show that if a graph is accepted, it must have diameter at most $d$.
         \vspace{10px}
     \newline
    Proof of $\Rightarrow:$ 
    
    The correctness of the diameter follows from the classical two-phase BFS algorithm for computing the diameter of a tree \cite{clrs2022}. Specifically, a BFS from an arbitrary vertex yields a farthest vertex $r_1$, and a subsequent BFS from $r_1$ returns a vertex $r_2$ maximizing the distance from $r_1$, which is equal to the diameter of the tree.

    We thus only need to show that for any vertex $v$, it holds that $d_1(v) = d(v,r_1)$ and $d_2(v) = d(v,r_2)$, and that the value $m_1(v)$ is the same for all vertices and equal $\max_u d_1(u)$, where $u$ is a vertex.  

    Let us start with the value of $m_1$. Assume there are two vertices $v,u$ such that $m_1(u)\neq m_1(v)$. Consider the path from $u$ to $v$. Somewhere on the path, there must be a vertex $x$, where the value $m_1(x) = m_1(u)$, but for the next vertex $y$ on the path $m_1(y) \neq m_1(x)$. The verifier $\mathcal{A}(x)$ would reject according to the condition 2. 

    If $m_1(v)<\max_{u\in V}(d(u,r_1))$, then there is a vertex $x$ such that $d(x,r_1) >m_1(v)$ and $\mathcal{A}(x)$ would reject according to the condition 3. 

    If $m_1(v)>\max_{u\in V}(d(u,r_1))$, then there is no vertex $x$ such that $d(x,r_1) = m_1(x)$ and thus also no vertex such that $d_1(x) = m_1(x)$. Because of that, there cannot be a vertex such that $d_2(x) = 0$, which means there is a vertex $y$ with no neighbor $z\in N(y)$ with $d_2(z) = d_2(y)-1$. The verifier $\mathcal{A}(y)$ would reject according to the condition 1(b).

    For the converse implication, we show that any yes-instance is always accepted.
    \vspace{10px}
    \newline
    Proof of $\Leftarrow:$

    The component $d_1$ is the distance from the origin vertex $r_1$. In trees, it holds that for any vertex with the exception of $r_1$ there always is one neighbor closer to $r_1$ and the others are further away.
    The same holds for $d_2$ and $r_2$, which represent the second run of BFS on the tree.
    Thus, the rejection criteria of condition 1 are not met at any vertex.

    The component $m_1$ is set to the maximum value of $d_1$. Each vertex receives the same values. 
    The rejection criteria of conditions 2, 3, and 4 are thus also not met.

    As the diameter of the tree is $d$, all the components $d_1,m_1, d_2$ can be represented with $\lceil\log_2(d+1)\rceil$ bits and hence the certification scheme is of size $3\lceil\log_2(d+1)\rceil$.

     We have thus shown that the 1-local certification scheme $(f,\mathcal{A})$ of size $3\lceil\log_2(d+1)\rceil$ accepts a tree if and only if it has diameter at most $d$.
\end{proof}

Our direct construction optimizes the hidden constants of general logic-based meta-theorems down to a precise multiplicative factor of $3$. This leaves an open question regarding the exact bit complexity of diameter verification on trees: What is the smallest constant $c$ such that $c \log_2 d + O(1)$ bits are both necessary and sufficient? Determining whether our upper bound of $3$ can be further reduced, or establishing a matching lower bound for anonymous networks, remains unresolved.

\bibliographystyle{plain}
\bibliography{ref}

\end{document}